\documentclass[11pt]{article}
\usepackage{amsmath,amssymb,amsbsy,amsfonts,
latexsym,amsopn,amstext,amsxtra,amsthm,fullpage}

\usepackage{hyperref}

\DeclareMathOperator{\erf}{erf}
\DeclareMathOperator{\poly}{poly}
\DeclareMathOperator{\diag}{diag}

\author{
Divesh Aggarwal\thanks{Courant Institute of Mathematical Sciences, New York
 University. Email: {\tt  divesha@cs.nyu.edu}.}
\and
Oded Regev\thanks{Courant Institute of Mathematical Sciences, New York University. This material is based upon work supported by the National Science Foundation under Grant No.~CCF-1320188. Any opinions, findings, and conclusions or recommendations expressed in this material are those of the authors and do not necessarily reflect the views of the National Science Foundation.}
}

\title{A Note on Discrete Gaussian Combinations of Lattice Vectors}

\begin{document}

\numberwithin{equation}{section}

\newtheorem{theorem}[equation]{Theorem}
\newtheorem{lemma}[equation]{Lemma}
\newtheorem{cor}[equation]{Corollary}
\newtheorem{prop}[equation]{Proposition}
\newtheorem{definition}[equation]{Definition}
\newtheorem{question}[equation]{Question}
\newtheorem{conj}[equation]{Conjecture}
\newtheorem{fact}[equation]{Fact}
 \newtheorem{observation}[equation]{Observation}
 \newtheorem{cla}[equation]{Claim}

\def\cA{{\mathcal A}}
\def\cB{{\mathcal B}}
\def\cC{{\mathcal C}}
\def\cD{{\mathcal D}}
\def\cE{{\mathcal E}}
\def\cF{{\mathcal F}}
\def\cG{{\mathcal G}}
\def\cH{{\mathcal H}}
\def\cI{{\mathcal I}}
\def\cJ{{\mathcal J}}
\def\cK{{\mathcal K}}
\def\cL{{\mathcal L}}
\def\cM{{\mathcal M}}
\def\cN{{\mathcal N}}
\def\cO{{\mathcal O}}
\def\cP{{\mathcal P}}
\def\cQ{{\mathcal Q}}
\def\cR{{\mathcal R}}
\def\cS{{\mathcal S}}
\def\cT{{\mathcal T}}
\def\cU{{\mathcal U}}
\def\cV{{\mathcal V}}
\def\cW{{\mathcal W}}
\def\cX{{\mathcal X}}
\def\cY{{\mathcal Y}}
\def\cZ{{\mathcal Z}}

\def\({\left(}
\def\){\right)}

\def \N{{\mathbb{N}}}
\def \K{{\mathbb{K}}}
\def \Z{{\mathbb{Z}}}
\def \Q{{\mathbb{Q}}}
\def \R{{\mathbb{R}}}
\def \C{{\mathbb{C}}}

\def \bx{{\mathbf{x}}}
\def \bc{{\mathbf{c}}}
\def \bu{{\mathbf{u}}}
\def \be{{\mathbf{e}}}
\def \bv{{\mathbf{v}}}
\def \by{{\mathbf{y}}}
\def \bb{{\mathbf{b}}}
\def \bz{{\mathbf{z}}}
\def \bw{{\mathbf{w}}}
\def \br{{\mathbf{r}}}
\def \bB{{\mathbf{B}}}

\def \bv{{\mathbf{v}}}

\def \cc{R}
\def \cl{r}
\def \dc{S}
\def \dl{s}

\def \eps{\varepsilon}

\maketitle

\begin{abstract}
We analyze the distribution of $\sum_{i=1}^m v_i \bx_i$ where $\bx_1,\ldots,\bx_m$ are fixed vectors from some lattice $\cL \subset \R^n$ (say $\Z^n$) and $v_1,\ldots,v_m$ are chosen independently from a discrete Gaussian distribution over $\Z$. We show that under a natural constraint on $\bx_1,\ldots,\bx_m$, if the $v_i$ are chosen from a wide enough Gaussian, the sum is statistically close to a discrete Gaussian over $\cL$. We also analyze the case of $\bx_1,\ldots,\bx_m$ that are themselves chosen from a discrete Gaussian distribution (and fixed).

Our results simplify and qualitatively improve upon a recent result by Agrawal, Gentry, Halevi, and Sahai~\cite{AGHS13}.
\end{abstract}

\newpage
\section{Introduction}
\label{sec:intro}

The study of discrete Gaussian distributions has been a crucial step in most of the work in lattice-based cryptography in the last few years. A discrete Gaussian distribution over some fixed lattice $\cL$, denoted as $\cD_{\cL, s}$ for some parameter $s>0$, is a distribution in which each lattice point is sampled with probability proportional to the probability density function of a continuous Gaussian distribution of width $s$ evaluated at that point.  There have been several results, e.g.~\cite{MR07,GPV08,Pei10,BF11,MP13}, showing that these functions have many useful properties similar to the corresponding continuous functions. However, some aspects of the discrete Gaussian distribution are not as well understood as those of its continuous counterpart.

\paragraph{Our Result.}

We analyze the following distribution.  Let $\bx_1, \ldots, \bx_m$ be $m$ fixed points in $\Z^n$ for some $m > n$. Let $X$ be the $n \times m$ matrix formed by the vectors $\bx_1, \ldots, \bx_m$ as the column vectors. For $r>0$, define the distribution
\[
\cE_{X, r} := \{ X \cdot \bv \: : \: \bv \leftarrow \cD_{\Z^m, r}\} \;.
\]
We show in Theorem~\ref{thm:sufficient} in Section~\ref{sec:sufficient} that if $X$ satisfies a certain constraint and if $r$ is large enough, then the distribution $\cE_{X,r}$ is statistically close to the discrete Gaussian distribution with appropriate covariance (denoted as $\cD_{\Z^n, rX^T}$). The constraint on $X$ that we need is that $X$ has small entries and additionally there are short vectors $\bu_1, \ldots, \bu_n \in \Z^m$ such that $X \cdot \bu_i = \be_i$ for $1\leq i \leq n$, where $\be_i$ are unit vectors in $n$ dimensions. 

We then show that if the vectors $\bx_1, \ldots, \bx_m$ are chosen according to the discrete Gaussian distribution, then $X$ satisfies the aforementioned constraint, and hence the following result holds:

\begin{theorem}[Informal]
\label{thm:main-informal}
For integers $n \ge 1$, $m = \poly(n)$, $s =  \Omega(\sqrt{\log (n/\eps)})$,  and $r = \tilde{\Omega}(n s \sqrt{\log{(1/\eps)}}) $, and $X \leftarrow (\cD_{\Z^n, s})^m$, we have with probability $1 - 2^{-n}$, the statistical distance between $\cE_{X,r}$ and $\cD_{\Z^n, r X^{T}}$ is at most $\eps$.
\end{theorem}

The formal statement and proof appear as Theorem~\ref{thm:main} in Section~\ref{sec:DG}. We also mention there how to obtain a more general version for the case in which the columns of $X$ are chosen from a discrete Gaussian distribution over an arbitrary lattice $\cL \subset \R^n$; this follows easily from the case of $\Z^n$.

The motivation for studying the distribution $\cE_{X,r}$ for this choice of $X$ comes from a recent breakthrough construction of multilinear maps from ideal lattices by Garg, Gentry, and Halevi~\cite{GGH13} (see~\cite{AGHS13} for details). Also, more recently, our result has been used by Ling et al.~\cite{LPSS13} to give a lattice-based traitor tracing scheme.

Our result improves on the main result of Agrawal, Gentry, Halevi, and Sahai~\cite{AGHS13} who showed a statement similar to Theorem~\ref{thm:main-informal} with  the bound on $r$ depending linearly on $m$. Improving this dependence on $m$ was left by them as an open question, which we answer here.

\paragraph{Our Techniques.}

As was done in~\cite{AGHS13}, the first part of our proof is to reduce the question to a question about bounding the successive minima of a so-called orthogonal lattice. Namely, let $A \subset \Z^m$ be the \emph{orthogonal lattice} of $X$, i.e., the lattice formed by all integer vectors in $\Z^m$ that are orthogonal to all $n$ rows of $X$,
which is (typically) a lattice of rank $m-n$. Our goal then is to find
$m-n$ linearly independent vectors in it whose maximum length is as small as possible, i.e., we need to bound the last successive minimum $\lambda_{m-n}$ of the lattice $A$. We feel that the techniques used in this part are rather standard, and instead of repeating the proofs, we cite the relevant lemmas from~\cite{AGHS13}.

     The main part of our proof (and that of~\cite{AGHS13}) is therefore bounding $\lambda_{m-n}(A)$. The approach we take is using a simple, yet powerful, idea from a recent result by Kuperberg, Lovett, and Peled~\cite{KLP12} (who were interested in an entirely different application, namely, showing the existence of some combinatorial structures like $t$-wise independent permutations). Using this idea, we get that in order to prove a bound on the successive minima of the orthogonal lattice, it suffices to satisfy two constraints: (1) that $X$ has small entries, and (2) that there are short vectors $\bu_1,\ldots,\bu_n$ such that $X \cdot \bu_i = \be_i$ for $1\leq i \leq n$, where $\be_i$ are unit vectors in $n$ dimensions.

     This already gives us sufficient constraints on $X$ and $r$ under which $\cE_{X,r}$ is guaranteed to be close to a discrete Gaussian. In order to complete the proof of Theorem~\ref{thm:main-informal}, it suffices to show that the $X$ chosen there satisfies these two constraints with high probability. The first constraint, namely, that $X$ has small entries follows easily by definition. Proving the second constraint is trickier and we show this in Section 4 using a careful application of the pigeon-hole principle.

\paragraph{Comparison with~\cite{AGHS13}.}
As mentioned above, our Theorem~\ref{thm:main-informal} is similar to the main 
result of~\cite{AGHS13}, the main difference being the quantitative improvement 
in the bound on $r$.
Namely, the bound on $r$ obtained by~\cite{AGHS13} is $\tilde{\Omega}(mn \log{(1/\eps)})$ which could be much worse than our bound for large $m$. We note that our result gives a worse bound if $s \gg m$. However, for applications like~\cite{GGH13,LPSS13}, one chooses $s =  \Theta(\sqrt{\log (n/\eps)})$ in which case our result gives a better bound for all $m,n,\eps$.

Another advantage of our approach is that our Theorem~\ref{thm:sufficient} (which we use to prove Theorem~\ref{thm:main-informal}) applies to a \emph{general} choice of $X$, and not just to a random $X$. This might be useful in future applications. 

In terms of techniques, as mentioned above, the most challenging step in both our proof and the proof in~\cite{AGHS13} is to bound the last successive minimum $\lambda_{m-n}$ of the orthogonal lattice $A$, which is a lattice of rank $m-n$ in dimension $m$. The way this is done in~\cite{AGHS13} is by first defining a superlattice $A_q$ of $A$ which is of \emph{full rank}, 
and then considering the dual lattice $M_q$ of $A_q$. They then obtain a lower bound on $\lambda_{n+1}(M_q)$, which using Banasczcyk's transference theorem~\cite{Ban93}, implies an upper bound on $\lambda_{m-n}(A_q)$. Finally, they argue that this is also an upper bound on $\lambda_{m-n}(A)$. This somewhat indirect proof results in the bound on $r$ depending linearly on $m$. Improving this dependence was left as an open question in~\cite{AGHS13}. In comparison, our method to bound $\lambda_{m-n}(A)$ is more direct. 

\paragraph{Other Related Work.}

Micciancio and Peikert~\cite{MP13} recently showed hardness results for the main lattice-based cryptographic problems. One of the key ingredients in the proof was a new convolution theorem which was a strengthening of a previous similar result by Peikert~\cite{Pei10}. This theorem, like our result, also looks at sums of discrete Gaussian samples. The difference is that in their statement the combination vector $\bv$ is fixed and the matrix $X$ comprises of (spherical) discrete Gaussian with parameter $s$ significantly bigger than the smoothing parameter of the underlying lattice. In our case, the matrix is fixed ``once and for all", and only $\bv$ varies. In our setting we essentially analyze the sum of $1$-dimensional discrete Gaussians in $n$-dimensional space. 

We note that by combining the result of~\cite{MP13} with that of~\cite{AGHS13}, it might be possible to derive a result similar to our Theorem~\ref{thm:main-informal}, i.e., to improve the dependence of $r$ on $m$. The idea would be to partition the sum $\sum_{i=1}^m v_i \bx_i$ into small blocks, argue that each block is close to a discrete Gaussian (using~\cite{AGHS13}) and then arguing that the overall sum must therefore also be close to a discrete Gaussian (using~\cite{MP13}). As this approach would almost certainly lead to worse parameters and a far more complicated proof, we did not attempt to pursue it.

\section{Preliminaries}
\label{sec:prelim}

All logarithms, unless otherwise stated, are to the base $2$. Natural logarithms, i.e., to the base $e$, are denoted by $\ln$. The norm $\| \cdot \|$ considered in this paper is the $\ell_2$ norm, unless otherwise stated. The transpose and inverse of a matrix $A$ are denoted as $A^T$ and $A^{-1}$, respectively. The inverse transpose of a matrix $A$ is denoted as $A^{-T}$.

Let $D$ be a discrete distribution. We denote by $D[x]$ the probability it assigns to $x$, and by $X \sim D$ a random variable distributed according to $D$. For two distributions $D,D'$ their statistical distance is $\Delta(D;D')=\tfrac{1}{2}\sum_x |D[x]-D'[x]|$.

\paragraph{Gaussian function.}

For any $s > 0$, the (spherical) Gaussian function on $\R^n$ with parameter $s$ is defined as $\rho_{s} (\bx) = \exp (- \pi \| \bx \|^2/s^2)$ for all $\bx \in \R^n$. If $s = 1$, then it is omitted. For a rank-$n$ matrix $S \in \R^{m \times n}$, the ellipsoidal Gaussian function on $\R^n$ with covariance matrix $\Sigma = S^T S$ is defined as:
\[
\forall \bx \in \R^n, \; \; \rho_{S}(\bx) = \exp\(-\pi \bx^T (S^T S)^{-1} \bx\right) \; .
\]
When $S = s I_n$, $\rho_{S}$ is the same as $\rho_{s}$. Also, if $m = n$, then $\rho_{S}(\bx) = \rho (S^{-T}\bx)$.

\paragraph{Lattices.}

A lattice is a discrete additive subgroup of $\mathbb{R}^m$. A set of linearly independent vectors that generates a lattice is called a basis and is denoted by $\bB = \{\bb_1, \ldots, \bb_n\} \subset \mathbb{R}^m$ for integers $m \ge n \ge 1$. The lattice generated by the basis $\bB$ is \[ \cL = \cL(\bB) = \Big\{ \bB \bz = \sum_{i=1}^n z_i \bb_i \: : \: \bz \in \mathbb{Z}^n \Big\} \; .\] We say that the rank of this lattice is $n$ and its dimension is $m$. For $i=1,\ldots,n$, the {\em successive minimum} $\lambda_i(\cL)$ is defined as the smallest value such that a ball of radius $\lambda_i(\cL)$ centered around the origin contains at least $i$ linearly independent lattice vectors. 

The dual lattice $\cL^*$ of $\cL$ is defined as $\cL^* = \{\bx \in \R^n \: : \: \langle \cL, \bx\rangle \subseteq \Z\}.$ For a lattice $\cL$ and positive real $\eps > 0$, the smoothing parameter $\eta_{\eps}(\cL)$ is the smallest real $s>0$ such that $\rho_{1/s}(\cL^* \setminus \{\mathbf{0}\}) \le \eps$, where $\rho_{1/s}(A)$ for a set $A$ denotes $\sum_{\bx \in A} \rho_{1/s}(\bx)$. 

For a rank-$n$ lattice $\cL$, $S \in \R^{m\times n}$, and $\bc \in \R^n$, the ellipsoidal Gaussian distribution with parameter $S$ and support $\cL + \bc$ is defined as:
\[
\forall \bx \in \cL + \bc, \; \cD_{\cL+\bc, S}(\bx) = \frac{\rho_{S}(\bx)}{\rho_{S}(\cL+\bc)} \; ,
\]
where $\rho_{S}(A)$ for a set $A$ denotes $\sum_{\bx \in A} \rho_{S}(\bx)$.

\paragraph{Matrices and Singular Values.}

For a rank-$n$ matrix $S \in \R^{m \times n}$, there exist orthogonal matrices $U, V$ (i.e., $U^T U = I$, and $V^T V = I$) and a diagonal matrix \[\Sigma = \diag(\sigma_1, \ldots, \sigma_n) \in \R^{m \times n}, \; \text{ with } \; \sigma_1 \ge \cdots \ge \sigma_n \ge 0 \;,\] such that $S = U\Sigma V^T$. The values $\sigma_1, \ldots, \sigma_n$ are called the singular values of $S$. The largest singular value of $S$ is denoted as $\sigma_1(S)$, and the least singular value is $\sigma_n(S)$.

\paragraph{Some Known Results.}

%The following is an easy generalization from Lemma 4.4 of ~\cite{MR07} to ellipsoidal Gaussian functions.

\begin{lemma}[{{\cite[Lemma 3]{AGHS13}}}] For any $c \ge 1/\sqrt{2\pi}$, rank-$n$ lattice $\cL$,  $0 < \eps < 1$,  and matrix $S$ s.t. $\sigma_n(S) \ge \eta_{\eps}(\cL)$, we have
\label{lem:short_vectors}
\[
\Pr_{\bv \leftarrow \cD_{\cL, S}} (\| \bv \| \ge \sigma_1(S) c \sqrt{n}) \le \frac{1 + \eps}{1 - \eps} (c \sqrt{2 \pi e} \cdot e^{-\pi c^2})^n \;.
\]
\end{lemma}

\begin{lemma}[{{\cite[Lemma 4]{AGHS13}}}] 
\label{lem:almost-unif}
For any rank-$m$ lattice $\cL$, $0<\eps<1$, vector $\bc \in \R^m$, and full-rank matrix $S \in \R^{m \times m}$, such that $\sigma_m(S) \ge \eta_{\eps}(\cL)$, we have \[\rho_{S}(\cL + \bc) \in \Big{[}\frac{1-\eps}{1 + \eps},1\Big{]} \cdot \rho_S(\cL) \; .\]
\end{lemma}

\begin{lemma}[{{\cite[Corollary 2]{AGHS13}}}]
\label{lem:random-distribution}
For any full-rank lattice $\cL \in \R^n$, $\eps > 0$, $c > 2$, and a rank-$n$ matrix $S$ such that $\sigma_n(S) \ge (1+ c) \eta_{\eps}(\cL)$, the following holds. For any $T \subset \cL$, and any $\bv \in \cL$, \[\cD_{\cL, S}(T) - \cD_{\cL, S}(T - \bv) \leq \frac{\erf(q/2 + 2q/c)}{\erf(2q)} \cdot \frac{1 + \eps}{1 - \eps} \;,\]
where $q = \frac{\|\bv\| \sqrt{\pi}}{\sigma_n(S)}$, and $\erf$ is the Gaussian error function defined as $\erf(x) = \frac{2}{\sqrt {\pi}}  \int_0^x e^{-t^2}\,dt. $
\end{lemma}

\begin{lemma} [{{\cite[Azuma's inequality, Chapter 7]{AS04}}}]
\label{lem:Azuma}
Let $X_0, X_1, \ldots$ be a set of random variables that form a discrete-time sub-martingale, i.e., for all $n\ge 0$, \[E[X_{n+1} \: | \: X_1, \ldots, X_n] \ge X_n \;. \] If for all $n \ge 0$, $|X_{n} - X_{n-1}| \le c$, then for all integers $N$ and positive real $t$, \[\Pr(X_N - X_0 \le - t) \le \exp\left(\frac{ -t^2}{2 Nc^2}\right) \; . \]
\end{lemma}

\section{Sufficient Condition for \texorpdfstring{$X$}{X}}
\label{sec:sufficient}

Let $X \in \Z^{n \times m}$ be a full row-rank matrix. For a full-rank matrix $R \in \R^{m\times m}$, and $\bc  \in \R^m$, we consider the distribution $\cE_{X, R, \bc }$ which is a generalization of $\cE_{X, r}$ defined in Section~\ref{sec:intro}.

\[
\cE_{X, R, \bc } := \{ X \cdot \bv \: : \: \bv \leftarrow \cD_{\Z^m + \bc, R}\} \;.
\]

In this section, we establish a sufficient condition on $X$ such that $\cE_{X, R, \bc }$ is statistically close to a discrete Gaussian distribution, namely, to $\cD_{\Z^n + X \bc, RX^T}$.

\begin{definition}\em
\label{def:quality}
Let $X$ be an $n \times m$ matrix and let $\bx_1, \ldots, \bx_n \in \Z^m$ be the rows of $X$. For any positive $q_1 = q_1(n), q_2 = q_2(n)$, the matrix $X$ is said to have {\em quality} $(q_1, q_2)$ if all column-vectors of $X$ have $\ell_2$ norm at most $ q_1$ and there exist pairwise orthogonal $\bu_1, \ldots, \bu_n \in \Z^m$, such that for all $i, j \in [n]$, $\bu_i \cdot \bx_j = \delta_{ij}$, and $\|\bu_i\|_2 \leq q_2$.
\end{definition}

Note that the fact that $X$ has quality $(q_1, q_2)$ for $q_1, q_2 \in \R^+$ implies that $X$ is a full row-rank matrix, and that $X \Z^m = \Z^n$.
In this section, we will prove the following theorem.

\begin{theorem}
\label{thm:sufficient}
For any $m> n\ge 1$, and $0<\eps < 1/3$, $\bc \in \R^m$, if $X$ is an $n \times m$ matrix of quality $(q_1, q_2)$, and $R \in \R^{m \times m}$ is full-rank such that \[\sigma_m(R) \ge (1 + q_1 q_2)  \cdot \sqrt{\frac{\ln(2(m-n)(1 + 1/\eps))}{\pi}}\; ,\] then
\[
\Delta(\cE_{X, R , \bc }, \: \cD_{\Z^n + X \bc, R X^T}) \le 2 \eps \; .
\]
\end{theorem}

Define a lattice $A = A(X)$ as the lattice containing integer vectors in $\Z^m$ orthogonal to all the row vectors of $X$,
\[
A = A(X) := \{ \bv \in \Z^m : X \cdot \bv = 0\} \;.
\]
Note that the rank of $A$ is $m-n$ since $X$ has rank $n$ and the vectors in $A$ span the vector space orthogonal to the $n$ row vectors of $X$. The following result is a slight extension of~{{\cite[Lemma 10]{AGHS13}}} (mentioned already in an unpublished version of~\cite{LPSS13}). It shows (in the case of centered spherical distribution) that if $r$ is bigger than the smoothing parameter of $A$, then $\cE_{X, r}$ is statistically indistinguishable from $\cD_{\Z^n, rX^T}$.

\begin{lemma}
\label{lem:AGHS-stat-dist}
Let $X$ be an $n \times m$ full row-rank matrix and let $A = A(X)$ be as defined above. For any full-rank matrix $R  \in \R^{m \times m}$, $\bc  \in \R^m$, such that $\sigma_m(R) > \eta_{\eps}(A)$ for some $\eps < 1/3$,
\[
\Delta(\cE_{X, R , \bc }, \: \cD_{X\Z^m + X\bc, R X^T }) \le 2 \eps \; .
\]
\end{lemma}
\begin{proof}
The support of $\cE_{X, R , \bc }$ is $X\Z^m + X\bc$. Fix some $\bz \in X\Z^m + X\bc$. The probability mass assigned to $\bz$ by $\cE_{X, R, \bc }$ is proportional to $\rho_R(A_\bz )=\rho(R^{-T}A_\bz )$, where
\[
A_{\bz} = A_{\bz}(X) := \{ \bv \in \Z^m + \bc : X \cdot \bv = \bz \} \;.
\]
Note that $A_{\bz} = A + \bw_\bz$, for any arbitrary element $\bw_\bz \in A_{\bz}$ (since $A_{\bz}$ is non-empty). 

Consider now the quantity $\rho(R^{-T}A_\bz )$. Let $X' = X \cdot R^T \in \R^{n \times m}$, and note that $\ker X' = R^{-T} \ker X$.
Let $Y' =  (X' X'^T)^{-1} X'$ be the pseudo-inverse of $X'$ (i.e., $X' Y'^T = I_n$ and the rows of $Y'$ span the same linear subspace of $\R^m$ as the rows of $X'$).

 Define  $\bu_{\bz}:= Y'^T \bz$. Note that $\bu_{\bz}$ is the point in the affine subspace $\ker X' + R^{-T} \bw_\bz $ closest to the origin: it is clearly in this affine subspace, and moreover, being in the row space of $Y'$, it must be
orthogonal to $\ker Y' = \ker X'$. Therefore, for any  $\br \in R^{-T}A_\bz $, we have that $\br - \bu_{\bz}$ is orthogonal to $\bu_{\bz}$, and thus, $\rho(\br) = \rho(\bu_{\bz})\cdot \rho(\br - \bu_{\bz})$. Hence, 
\[
\rho(R^{-T}A_\bz) = \rho(\bu_{\bz}) \cdot \rho(R^{-T}A_\bz -\bu_{\bz } ) \;.
\]
Since $\bu_{\bz} \in R^{-T} \ker X + R^{-T} \bw_\bz $, it follows that \[R^{-T}A_\bz -\bu_{\bz } = R^{-T}(A - \bc')\;, \] for some $\bc'$ in the span of the lattice $A$. Thus, using Lemma~\ref{lem:almost-unif}, we have that 
\begin{eqnarray*}
\rho(R^{-T}A_\bz ) &=& \rho(\bu_{\bz}) \cdot \rho_{R}(A - \bc') \\
                   &\in&  \Big{[}\frac{1-\eps}{1 + \eps},1\Big{]} \cdot \rho_{R}(A) \cdot \rho(\bu_{\bz}) \\
                            &=& \Big{[}\frac{1-\eps}{1 + \eps},1\Big{]} \cdot \rho_{R}(A) \cdot \rho(X'^T (X' X'^T)^{-1}  \bz) = \Big{[}\frac{1-\eps}{1 + \eps},1\Big{]} \cdot \rho_{R }(A) \cdot \rho_{X'^T}(\bz) \;. 
\end{eqnarray*}
 This implies that the statistical distance between $\cE_{X, R , \bc }$ and $\cD_{X\Z^m + X\bc, RX^T}$ is at most $1 - \frac{1 - \eps}{1 + \eps} \le 2 \eps$. 
\end{proof}

Thus, in order to bound the required statistical distance, we need to bound the smoothing parameter $\eta_{\eps}(A)$, for which we use the following result.

\begin{lemma}[{{\cite[Lemma 3.3]{MR07}}}]
\label{lem:smoothing-param}
For any rank-$n$ lattice $\cL$ and any $\eps > 0$, 
\[\eta_{\eps}(\cL) \le \lambda_n(\cL) \cdot \sqrt{\frac{\ln(2n(1 + 1/\eps))}{\pi}} \;.\]
In particular, $\eta_{\eps}(\Z^n) \le \sqrt{\frac{\ln(2n(1 + 1/\eps))}{\pi}}$.
\end{lemma}

Therefore, we now only need to bound $\lambda_{m-n}(A)$, and then using Lemma~\ref{lem:smoothing-param}, we get a bound on $\eta_{\eps}(A)$. For this, we use the following result from ~\cite{KLP12}, which we state and prove in our notation.

\begin{lemma}
\label{lem:KLP}
Let $X \in \Z^{n \times m}$ and let $\bx_1, \ldots, \bx_n \in \Z^m$ be the rows of $X$. If $X$ has quality $(q_1, q_2)$, then there exist linearly independent $\bv_1, \ldots, \bv_{m-n} \in \Z^m$ such that  $\bv_k \cdot \bx_j = 0$  and $\|\bv_k\|_2 \leq 1 + q_1 q_2$ for all $j, k$. In particular $\lambda_{m-n}(A) \leq 1 + q_1 q_2$.
\end{lemma}
\begin{proof}
Let $\bu_1, \ldots, \bu_n$ be as in Definition~\ref{def:quality}.
Define vectors $\bv_1, \ldots, \bv_m$ as:
\[
\bv_k = \be_k - \sum_{i = 1}^n x_{ik} \bu_i \; ,
\]
where $x_{ik}$ is the $k$-th coordinate of $\bx_i$, and $\be_k$'s are  standard unit vectors in $m$ dimensions. Thus, \[\bv_k \cdot \bx_j = x_{jk} - \sum_{i = 1}^n x_{ik} (\bu_i \cdot \bx_j) = x_{jk} - x_{jk} = 0 \;,\]
using the fact that $\bu_i \cdot \bx_j = \delta_{ij}$. 

Using the fact that $\bu_i$'s are orthogonal and the triangle inequality, we have that 
\[ \|\bv_k\| \le 1 + \sqrt{\sum_{i=1}^n |x_{ik}|^2 \|\bu_i\|^2} \le 1 + q_2 \sqrt{ \sum_{i=1}^n |x_{ik}|^2} \le 1 + q_1 q_2 \; .\]

Also, clearly there are at least $m - n$ linearly independent vectors in $\bv_1, \ldots, \bv_m$, since, by definition, the set of vectors $\{\bv_1, \ldots, \bv_m, \bu_1, \ldots, \bu_n\}$ together generate all the unit vectors $\be_k$, and hence span $\R^m$.
\end{proof} 

 Combining Lemmata~\ref{lem:AGHS-stat-dist},~\ref{lem:smoothing-param}, and~\ref{lem:KLP} immediately leads to a proof of Theorem~\ref{thm:sufficient}.

We note that one can relax Definition~\ref{def:quality} by not requiring the pairwise orthogonality of the vectors $\bu_1, \ldots, \bu_n$. In this case, Theorem 3.2 still holds if we take $q_1$ to be a bound on the $\ell_1$ norm of the columns of $X$ (instead of the $\ell_2$ norm as it is now).

\section{The matrix \texorpdfstring{$X$}{X} is of good quality}
\label{sec:quality}

In this section, we show that $X \leftarrow (\cD_{\Z^n, S})^m$ has quality $(q_1, q_2)$ for ``small" $q_1, q_2$ with overwhelming probability. For proving this result, we need the following claim that easily follows from the pigeon-hole principle.

\begin{cla}
\label{claim:pigeon-hole}
Let $B, n$ be any positive integers such that $Bn \ge 16$, and let $\bx_1, \bx_2, \ldots \in \Z^n$ be such that $\|\bx_j\|_{\infty} \leq B$ for $1 \le j \le \lfloor 2n \log (B n) \rfloor$. Then there exist $\alpha_1, \ldots, \alpha_{\lfloor 2n  \log Bn \rfloor} \in \{-1,0,1\}$, not all zero, such that $\sum_{j = 1}^{\lfloor 2n \log B n \rfloor} \alpha_j \bx_j = 0$.
\end{cla}
\begin{proof}
Let $\ell:= \lfloor 2n  \log Bn \rfloor$. Any $0/1$ combinations of $\bx_1, \ldots, \bx_{\ell}$ is in $\{-B\ell, \ldots, -1, 0, 1, \ldots, B \ell \}^n$, and thus the total number of distinct resulting vectors is at most $(2B\ell + 1)^n < 2^{\ell}$. By the pigeon-hole principle there exist two distinct $0/1$ combinations that result in the same vector, which implies the result by taking their difference.
\end{proof}

\begin{lemma}
\label{lem:main}
Let $n \ge 100$ be an integer, and let $\eps = \eps(n) \in (0, 1/1000)$. Let $S \in \R^{n \times n}$ be such that $\sigma_n(S) \ge 9 \sqrt{\frac{\ln(2n(1 + 1/\eps))}{\pi}}$. Denote $\sigma_1 = \sigma_1(S)$ and $\sigma_n = \sigma_n(S)$. Let $m$ be an integer such that $m > 30 n \log (\sigma_1 n)$ and let $X \leftarrow (\cD_{\Z^n, S})^m$ be an $n \times m$ matrix. Then, with probability $1 - 2^{-n}$, $X$ has quality $(\sigma_1 \sqrt{n \log m} , 2 \sqrt{30 n \log (\sigma_1 n)})$.
\end{lemma}
\begin{proof} By Lemma \ref{lem:short_vectors}, with $c = \sqrt{\log m}$, we have using the union bound that with probability $1 - 2^{-n-1}$, all column-vectors of $X$ have norm at most $\sigma_1 \sqrt{n \log m}$. It is sufficient to show that with probability $1 - 2^{-n-1}$, there exist pairwise orthogonal $\bu_1, \ldots, \bu_n \in \Z^m$ such that $\|\bu_i\|_2 \leq 2 \sqrt{30 n \log (\sigma_1 n)} $, and $X \cdot \bu_i = \be_i$ for $1\leq i \leq n$.

Denote $t: = 3 n \log (\sigma_1 n)$. For all $i \in [n]$, we show the existence of $\bu_i$ with probability $1 - 2^{-t/10}$, and then the result follows by the union bound.
Let $\bx_1, \ldots, \bx_m$ be the columns of $X$. Define sets $\cS_j \subset \Z^n$ for $1 \le j \le m$ as follows:
\[
\forall j \in [m], \; \cS_j = \Big{\{} \sum_{k = 1}^j \alpha_k \bx_k \: : \: \alpha_k \in \{-1,0,1\}\Big{\}} \; .
\]

From now on we show the existence of $\bu_1$ and at the end of the proof we explain how to show the existence of the other $\bu_i$. We will show that with probability $1 - 2^{-t/10}$, $\cS_{\lfloor 10 t \rfloor }  \cap (\cS_{\lfloor 10 t \rfloor }  + \be_1)$ is non-empty, which implies that $\be_1$ can be obtained as a $\{-2,-1,0,1,2\}$ combination of the first $ \lfloor 10 t \rfloor$ vectors. Then we are done since the vector $\bu_1$ can be defined as this coefficient vector with first $\lfloor 10 t \rfloor$ coordinates in $\{-2, -1, 0, 1, 2\}$ and the remaining coordinates zero. It is easy to see that $\|\bu_1\|$ is at most $2 \sqrt{10 t}$.

The idea used to prove this is that if this does not happen, then $\cS_j$ and $\cS_j + \be_1$ must be disjoint for all $j \le \lfloor 10 t \rfloor$, and so at most one of $\bx_{j+1}$ and $\bx_{j+1} - \be_1$ is in $\cS_j$ for any $j \le \lfloor 10 t \rfloor$. Using the randomness of $\bx_j$, we conclude that with constant probability, $\bx_{j+1} \notin \cS_j$ and this cannot happen for many values of $j$ by Claim~\ref{claim:pigeon-hole}.

Define binary random variables $Y_1, \ldots, Y_m$ as follows:
$Y_j = 1$ if and only if either 
\begin{itemize}
\item $\cS_j \cap (\cS_j + \be_1)$ is non-empty, or
\item $\bx_{j+1} \notin \cS_j$ and $\|\bx_{j+1}\|_\infty \le \sigma_1 \sqrt n$.
\end{itemize}
Consider the following claims about the distribution of $Y_1, \ldots, Y_m$.

\begin{cla}
\label{claim:main1}
 For any values of $y_1, \ldots, y_{j-1}$, and for all $j$, $\Pr(Y_j = 1 \; | \; Y_1 = y_1, \ldots, Y_{j-1} = y_{j-1}) \geq 0.3$.
\end{cla}
\begin{proof}
Throughout this proof, we condition on arbitrary values of $\bx_1,\ldots,\bx_j$, and our goal is to show that the probability of $Y_j=1$ is at least $0.3$.

If $\cS_j \cap (\cS_j + \be_1)$ is non-empty (an event that depends only on $\bx_1,\ldots,\bx_j$), then we are done. Assume therefore that our choice of $\bx_1,\ldots,\bx_j$ is such that $\cS_j$ and $\cS_j+\be_1$ are disjoint. By Lemma~\ref{lem:short_vectors}, with all but exponentially small probability in $n$, $\|\bx_{j+1}\|_\infty \le \|\bx_{j+1}\|_2 \le \sigma_1\sqrt{n}$. It therefore suffices to show that $\bx_{j+1} \in \cS_j$ with probability at most, say, $0.695$. This follows since from Lemma~\ref{lem:smoothing-param}, we get that $\sigma_n \ge 9 \eta_{\eps}(\Z^n)$, and thus we have by Lemma~\ref{lem:random-distribution} that for any set $\cS \subset \Z^n$
\begin{equation*}
%\label{eq:random}
\Pr(\bx_{j+1} \in \cS) - \Pr(\bx_{j+1} \in \cS + \be_1) \le \frac{\erf(3 \sqrt{\pi}/(4\sigma_n))}{\erf(2 \sqrt{\pi}/\sigma_n)} \cdot \frac{1 + \eps}{1 - \eps} < 0.39 \;,
\end{equation*}
where the last inequality follows because $\erf$ is nearly linear for small arguments. Finally, the claim follows by noticing that the sum of the two probabilities in the left hand side is at most $1$.
\end{proof}

\begin{cla}
\label{claim:main2}
\[\Pr(Y_1 + \cdots + Y_{\lfloor 10 t \rfloor } \geq  2 t )   \ge 1 - 2^{-t/10}\;.\] 
\end{cla}
\begin{proof}
Define random variables $Z_0, Z_1, \ldots$ recursively as $Z_0 = 0$ and for $i \ge 1$, $Z_i = Z_{i-1} + Y_i - \frac{3}{10}$, so for all $j$,  \[ Y_1 + \cdots + Y_j  = Z_j - Z_0 + \frac{3j}{10} \;. \] Clearly $Z_0, Z_1, \ldots$ is a sub-martingale, and for all $j \ge 1$, $|Z_{j} - Z_{j-1}| \le \frac{7}{10}$.  Thus, using Lemma \ref{lem:Azuma}, we get that 
\[\Pr(Y_1 + \cdots + Y_{\lfloor 10 t \rfloor } \leq  2 t ) = \Pr(Z_{\lfloor 10 t \rfloor } - Z_0 \leq -t) \le \exp\left(\frac{ -t^2}{2\lfloor 10 t \rfloor (0.7)^2}\right) \le 2^{-t/10} \;, \]
which implies the result. 
\end{proof}

\begin{cla}
\label{claim:main3}
If $Y_1 + \cdots + Y_{\lfloor 10 t  \rfloor } \ge  2 t$, then $\cS_{\lfloor 10 t \rfloor } \cap (\cS_{\lfloor 10 t \rfloor }  + \be_1)$ is non-empty.
\end{cla}
\begin{proof}
Assume $Y_{j_1} = \cdots = Y_{j_{\lfloor 2t\rfloor }} = 1$ and $\cS_{\lfloor 10 t \rfloor } \cap (\cS_{\lfloor 10 t \rfloor }  + \be_1)$ is empty. Consider the vectors $\bx_{j_1 + 1}, \ldots, \bx_{j_{\lfloor 2t \rfloor} + 1}$. We have that $\|\bx_{j_k + 1}\|_{\infty} \le \sigma_1 \sqrt n$ for $1 \le k \le \lfloor 2t  \rfloor $. Since $2t = 6n \log (\sigma_1  n) \ge 2n \log (\sigma_1   n^{3/2})$, by Claim~\ref{claim:pigeon-hole}, there exist $\alpha_1, \ldots, \alpha_{\lfloor 2t \rfloor} \in \{-1,0,1\}$, not all zero, such that \[\sum_{k = 1}^{\lfloor 2t \rfloor} \alpha_k \bx_{j_k + 1} = 0 \;.\] Let $k$ be the largest index such that $\alpha_k \neq 0$. This implies $\bx_{j_k + 1} \in \cS_{j_k}$, contradicting the fact that $Y_{j_k} = 1$. 
\end{proof}

Claim~\ref{claim:main1}, \ref{claim:main2}, and \ref{claim:main3} together prove the existence of $\bu_1$. To obtain a vector $\bu_2$ that is orthogonal to $\bu_1$, we do the following. Let $X'$ be the $(n+1)\times m$ matrix formed by adding another row, $\bu_1$, to $X$. We follow the same argument as above with $X$ replaced by $X'$ and $\be_1$ replaced by $\be_2$. Let $\bx_j'\in \Z^{n+1}$ be $\bx_j$ concatenated with $u_{1j}$, the $j$-th coordinate of $\bu_1$. By construction $\|\bu_1\|_{\infty} \le 2$, and hence $\|\bx_j'\|_{\infty} \le \sigma_1 \sqrt n$ if and only if $\|\bx_j\|_{\infty} \le \sigma_1 \sqrt n$. For completing the proof, we need to bound 
$\Pr(\bx'_{j+1} \in \cS' ) - \Pr(\bx'_{j+1} \in \cS'  + \be_2)$
for an arbitrary set $\cS'  \subset \Z^{n+1}$. Since the $(n+1)$-th coordinate $\bx'_{j+1}$ is fixed to be $u_{1 (j+1)}$, we have that 
\[\Pr(\bx'_{j+1} \in \cS') - \Pr(\bx'_{j+1} \in \cS'  + \be_2) = \Pr(\bx_{j+1} \in \cS) - \Pr(\bx_{j+1} \in \cS + \be_2) \;, \]
where $\cS \subset \Z^n$ is the set of vectors formed by projecting the vectors of $\cS'' = \{\by \in \cS \: : \: y_{n+1} = u_{1 (j+1)}\}$ onto the first $n$ coordinates. The remainder of the proof is exactly the same. We similarly get $\bu_3, \ldots, \bu_n$, such that for all $i\le n$, $\bu_i$ is orthogonal to $\bu_1, \ldots, \bu_{i-1}$.
\end{proof}

\section{Sum of Discrete Gaussian Distribution}
\label{sec:DG}

In this section, we state the main result of this paper, the proof of which follows easily from the results proved in Section~\ref{sec:sufficient} and Section~\ref{sec:quality}.

\begin{theorem}
\label{thm:main}
Let $m> n \ge 100$ be integers, and let $\eps = \eps(n) \in (0, 1/1000)$. Let $R  \in \R^{m \times m}$ be a full-rank matrix, and let $\bc  \in \R^m$. Let $S \in \R^{n \times n}$ be a full-rank matrix, and let $X \leftarrow (\cD_{\Z^n,S})^m$. If $m \ge 30 n \log (\sigma_1(S) n )$, $\sigma_m(\cc ) \ge  10 n \sigma_1(S) \log m \sqrt{ \log (1/\eps) \log (n \sigma_1(S))}$, and $\sigma_n(S) \ge 9 \sqrt{\frac{\ln(2n(1 + 1/\eps))}{\pi}}$, then, with probability $1 - 2^{-n}$ over the choice of $X$, we have that \[\Delta(\cE_{X, \cc , \bc }, \: \cD_{\Z^n + X \bc, R X^T}) \le 2 \eps \; .\]
\end{theorem}
\begin{proof}
Theorem~\ref{thm:main} follows easily from Lemma~\ref{lem:main} and Theorem~\ref{thm:sufficient} using the observation that 
\begin{eqnarray*}
\sigma_m(\cc ) &\ge& 10 n \sigma_1(S) \log m \sqrt{ \log (1/\eps) \log (n \sigma_1(S))} \\
 &\ge& \Big{(}1 + 2 \sigma_1(S) \sqrt{n \log m}  \sqrt{30 n \log (\sigma_1(S) n)} \Big{)} \cdot \sqrt{\frac{\ln(2(m-n)(1 + 1/\eps))}{\pi}} \;.
\end{eqnarray*}
\end{proof}

It is easy to see that if we sample a vector $x$ from $D_{\Z^n,S}$ then $Bx$ is distributed like $D_{\cL,S B^T}$ where $\cL=\cL(B)$ is the lattice generated by the columns of $B$ (see, e.g., {{\cite[Fact 2]{AGHS13}}}).
As a result, we can extend Theorem~\ref{thm:main} to the case in which $X$ is sampled from $(D_{\cL,M})^m$ where $\cL=\cL(B)$ for some basis $B$, and $M=S B^T$. The only change is that the statistical closeness is to the distribution $D_{\cL + X \bc, RX^T}$.

\paragraph{Acknowledgements:} We thank Damien Stehl\'e for useful comments.

 \bibliographystyle{alphaabbrvprelim}
\bibliography{fuzzy}

\begin{thebibliography}{AGHS13}
\expandafter\ifx\csname urlstyle\endcsname\relax
  \providecommand{\doi}[1]{doi:\discretionary{}{}{}#1}\else
  \providecommand{\doi}{doi:\discretionary{}{}{}\begingroup
  \urlstyle{rm}\Url}\fi

\bibitem[AGHS13]{AGHS13}
S.~Agrawal, C.~Gentry, S.~Halevi, and A.~Sahai.
\newblock Discrete {G}aussian leftover hash lemma over infinite domains.
\newblock In \emph{ASIACRYPT}, pages 97--116. 2013.

\bibitem[AS04]{AS04}
N.~Alon and J.~Spencer.
\newblock \emph{The Probabilistic Method}.
\newblock Wiley, 2004.

\bibitem[Ban93]{Ban93}
W.~Banaszczyk.
\newblock New bounds in some transference theorems in the geometry of numbers.
\newblock \emph{Mathematische Annalen}, 296(01):625--635, 1993.

\bibitem[BF11]{BF11}
D.~Boneh and D.~Freeman.
\newblock Homomorphic signatures for polynomial functions.
\newblock In \emph{Advances in Cryptology-EUROCRYPT}, pages 149--168. Springer,
  2011.

\bibitem[GGH13]{GGH13}
S.~Garg, C.~Gentry, and S.~Halevi.
\newblock Candidate multilinear maps from ideal lattices.
\newblock In \emph{Advances in Cryptology-EUROCRYPT}, pages 1--17. Springer,
  2013.

\bibitem[GPV08]{GPV08}
C.~Gentry, C.~Peikert, and V.~Vaikuntanathan.
\newblock Trapdoors for hard lattices and new cryptographic constructions.
\newblock In \emph{40th Symposium on Theory of Computing-STOC}, pages 197--206.
  ACM, 2008.

\bibitem[KLP12]{KLP12}
G.~Kuperberg, S.~Lovett, and R.~Peled.
\newblock Probabilistic existence of rigid combinatorial structures.
\newblock In \emph{44th Symposium on Theory of Computing-STOC}, pages
  1091--1106. ACM, 2012.

\bibitem[LPSS13]{LPSS13}
S.~Ling, D.~H. Phan, D.~Stehl\'e, and R.~Steinfeld.
\newblock Hardness of k-lwe and applications in traitor tracing, 2013.
\newblock Unpublished manuscript.

\bibitem[MP13]{MP13}
D.~Micciancio and C.~Peikert.
\newblock Hardness of {SIS} and {LWE} with small parameters.
\newblock In \emph{Advances in Cryptology-CRYPTO}. Springer, 2013.

\bibitem[MR07]{MR07}
D.~Micciancio and O.~Regev.
\newblock Worst-case to average-case reductions based on {Gaussian} measures.
\newblock \emph{SIAM Journal on Computing}, 37(01):267--302, 2007.

\bibitem[Pei10]{Pei10}
C.~Peikert.
\newblock An efficient and parallel {Gaussian} sampler for lattices.
\newblock In \emph{Advances in Cryptology-CRYPTO}, pages 80--97. Springer,
  2010.

\end{thebibliography}

\end{document}